\documentclass[12pt,english]{article}
\usepackage[T1]{fontenc}
\usepackage[latin9]{inputenc}
\usepackage{geometry}
\usepackage{float}
\usepackage{amsmath}
\usepackage{amsthm}
\usepackage{amssymb}
\usepackage{setspace}
\usepackage{esint}
\usepackage[authoryear]{natbib}
\makeatletter

\providecommand{\tabularnewline}{\\}

\theoremstyle{plain}
\newtheorem{assumption}{\protect\assumptionname}
\theoremstyle{definition}
\newtheorem{defn}{\protect\definitionname}
\theoremstyle{plain}
\newtheorem{lem}{\protect\lemmaname}
\theoremstyle{plain}
\newtheorem{prop}{\protect\propositionname}
\theoremstyle{plain}
\newtheorem{thm}{\protect\theoremname}
\theoremstyle{plain}
\newtheorem{cor}{\protect\corollaryname}

\usepackage{amsfonts}\usepackage{babel}
\providecommand{\U}[1]{\protect\rule{.1in}{.1in}}
 \usepackage{fancyheadings}
\usepackage{setspace,graphicx}
\usepackage{pstricks-add,pst-all,pst-plot}
\usepackage{hyperref}
\usepackage{natbib}
\usepackage{xcolor}
\hypersetup{
colorlinks,
linkcolor={red!50!black},
citecolor={blue!50!black},
urlcolor={blue!80!black}
}
\usepackage{caption}
\usepackage{subcaption}

\makeatother

\usepackage{babel}
\providecommand{\assumptionname}{Assumption}
\providecommand{\corollaryname}{Corollary}
\providecommand{\definitionname}{Definition}
\providecommand{\lemmaname}{Lemma}
\providecommand{\propositionname}{Proposition}
\providecommand{\theoremname}{Theorem}

\begin{document}

\title{The Industry Supply Function and the Long-Run Competitive Equilibrium
with Heterogeneous Firms\thanks{We thank Gian Luca Clementi, Ernesto Dal Bó, Daniel Gottlieb,  Andrés
Rodriguez-Claire, the co-editor, and two anonymous referees for helpful
comments. Esponda: Department of Economics, UC Santa Barbara, 2127
North Hall, Santa Barbara, CA 93106, iesponda@ucsb.edu; Pouzo: Department
of Economics, UC Berkeley, 530 Evans Hall \#3880, Berkeley, CA 94720,
dpouzo@econ.berkeley.edu.} \bigskip{}
}

\author{%
\begin{tabular}{cc}
Ignacio Esponda~~~ & ~~~~~~~~Demian Pouzo\tabularnewline
(UC Santa Barbara)~~~ & ~~~~~~~~(UC Berkeley)\tabularnewline
\end{tabular}}
\maketitle
\begin{abstract}
In developing the theory of long-run competitive equilibrium (LRCE),
Marshall (1890) used the notion of a representative firm. The identity
of this firm, however, remained unclear. Subsequent theory either
focused on the case where all firms are identical or else incorporated
heterogeneity but disregarded the notion of a representative firm.
Using Hopenhayn's (1992) model of competitive industry dynamics, we
extend the theory of LRCE to account for heterogeneous firms and show
that the long-run supply function can indeed be characterized as the
solution to the minimization of a representative average cost function.
\end{abstract}
\bigskip{}

Keywords: long-run competitive equilibrium, representative firm\bigskip{}

\thispagestyle{empty}

\newpage{}

\thispagestyle{empty}

\tableofcontents{}\newpage{}

\setcounter{page}{1}

\section{Introduction}

The theory of long-run competitive equilibrium (LRCE), first developed
by Marshall in his \emph{Principles of Economics} (1890), has had
a profound influence on our understanding of competitive markets.\nocite{marshall1890principles}
One distinguishing feature of Marshall's theory is his conceptualization
of the (long-run) industry supply function. \citet{pigou1928analysis},
\citet{viner1953cost}{[}1931{]} and others subsequently formalized
Marshall's notion of LRCE. The latter author, in particular, is credited
for popularizing the typical diagram taught in introductory courses
and reproduced in Figure \ref{fig:txtbook}.

The figure represents an industry with fixed input prices where all
firms are identical and characterized by the marginal (MC) and average
(AC) cost functions depicted in the left panel. In an LRCE, price
is at the minimum point of the AC function, $p^{e}$, and aggregate
quantity is given by the demand function evaluated at that price,
$Q_{0}^{e}$. Suppose that there is a shift of the (inverse) demand
function from $P_{0}^{d}$ to $P_{1}^{d}$ in Figure \ref{fig:txtbook}.
In the short run, the number of firms stays fixed, so price and quantity
increase from the original LRCE at point $A$ to the new short-run
equilibrium at point $B$, a movement occurring along the short-run
supply function $S_{0}$. But then firms make positive (economic)
profits, and these profits attract additional firms into the market.
In the long-run, the new LRCE is at point $C$, where all firms make
zero profits at price $p^{e}$ and aggregate production increases
to $Q_{1}^{e}$. Thus, the (long-run) industry supply function, $S_{LR}$,
is horizontal at the minimum of the average cost function, $p^{e}$.

A distinguishing characteristic of Marshall's analysis is the notion
of a \emph{representative firm}. While Marshall recognized that there
are \emph{different} firms in an industry, subsequent developments
have focused on the case where all firms are identical in a long-run
equilibrium. \citet{viner1953cost}{[}1931{]}, pg. 222, justifies
this view:
\begin{quote}
``If there are particular units of the factors which retain permanently
advantages in value productivity over other units of similar factors,
these units, if hired, will have to be paid for in the long-run at
differential rates proportional to their value productivity, and if
employed by their owner should be charged for costing purposes with
the rates which could be obtained for them in the open market and
should be capitalized accordingly.''
\end{quote}
\bigskip{}

\begin{figure}[H]
\begin{centering}
\newpsobject{showgrid}{psgrid}{subgriddiv=1,griddots=10,gridlabels=7pt,gridcolor=red} 	
	\psset{unit=6mm} 	
	\psset{algebraic=true} 	
	\begin{pspicture}(0,0)(23,10) 
\psline[linewidth=0.05]{->}(0,0)(10,0) 	
	\psline[linewidth=0.05]{->}(0,0)(0,10) 	
		\rput(9,-0.5){quantity} 	
	\rput(-1.2,9){price \&} 	
	\rput(-1.1,8.5){costs}			 	
\psline(0,0)(8,7) 	
	\rput(8,7.5){{ $MC$}} 	
	\pscurve(2,8)(4,4)(9,6) 	
	\rput(9.6,6){{$AC$}}  	
	\psline[linestyle=dotted](0,3.85)(12,3.85)		 
		\psline[linestyle=dotted](4.45,0)(4.45,3.85)		 	
		\rput(4.45,-0.5){\Large{ $q^{e}$}}  
		\rput(-0.65,3.85){\Large{ $p^{e}$}}   	
	
	\psline[linewidth=0.05]{->}(12,0)(22,0) 	
	\psline[linewidth=0.05]{->}(12,0)(12,10) 	
		\rput(11.2,9){price} 	
		\psline[linestyle=solid](12,3.85)(22,3.85)

	\psline(14.3,2)(20,7) 	
	\rput(20,7.5){$S_{0}$} 	

\pscurve(14,8.1)(16,4.1)(21,2.1) 	
	\rput(13.5,8.5){$P^{d}_{0}$}  	
		\pscurve[linestyle=dashed](14,11)(16,6)(21,3.2) 	
	\rput(13.5,11){$P^{d}_{1}$}  	
	\psdot*[dotsize=5pt](16.43,3.85) 
		\psdot*[dotsize=5pt](19.43,3.85) 	
	\psdot*[dotsize=5pt](17.5,4.85)	 	
	\rput(16.43,4.35){$A$}      	
	\rput(17.5,5.35){$B$}   	
	\rput(19.43,4.35){$C$}    
	
\psline[linestyle=dotted](16.43,0)(16.43,3.85)  	
	\psline[linestyle=dotted](19.43,0)(19.43,3.85)  	

	\rput(22.6,3.85){$S_{LR}$}    
		\rput(16.43,-0.5){$Q^{e}_{0}$}  	
	\rput(19.43,-0.5){$Q^{e}_{1}$}  	
	\rput(22.1,-0.5){aggregate}  	
	\rput(22.1,-1.0){~quantity}

		\end{pspicture} 	
\par\end{centering}
\bigskip{}

\bigskip{}

\caption{\label{fig:txtbook}Textbook model of long-run competitive equilibrium.\protect \\
\emph{\footnotesize{}The left panel shows the average and marginal
cost curves of each identical firm. The right panel shows the industry
equilibrium. Profits are initially zero at equilibrium point A. As
demand shifts from $P_{0}^{d}$ to $P_{1}^{d}$, there is a short-run
change from A to B, price goes up, and profits become positive. Positive
profits drive entry into the industry, and, as more firms enter, the
industry moves to long-run equilibrium C, where profits are back to
zero and there is no further entry. In particular, the long-run supply
function is horizontal at the price $p^{e}$ that minimizes the AC
curve, which also happens to be where the MC and AC curves intersect.}}
\end{figure}

\bigskip{}
Viner's argument may justify why firms do not make rents in the presence
of markets that bid up the price of advantageous factors, such as
exceptional managerial ability. But the argument does not imply that
firms with different technologies or productivities cannot coexist
in equilibrium. A realistic feature of an industry is that low-productivity
firms can potentially become high-productivity firms and vice versa.
This feature implies that equilibrium will be characterized both by
coexistence of heterogeneous firms and turnover (entry and exit),
and it does not seem appropriate to exclude these realistic features
from a theory of LRCE.

Our objective in this paper is to go back to Marshall's original motivation
and to extend the classical theory of LRCE to the case of heterogeneous
firms. Fortunately, we don't have to formulate a new model, since
\citet{hopenhayn1992entry} actually introduced and studied a model
of competitive industry dynamics where firms' productivities evolve
over time and exit and entry is an equilibrium phenomenon. We take
the steady-state equilibrium in Hopenhayn's model as the natural extension
of the theory of LRCE to the case with heterogeneous firms. \citet{hopenhayn1992entry},
however, did not link his work to the early theory on LRCE, and our
contribution is to fill-in this gap.

Our main result is that the (long-run) industry supply function with
heterogeneous firms can indeed be characterized as the solution to
the minimization of a representative average cost function, as Marshall
originally envisioned. The standard textbook case, depicted in Figure
\ref{fig:txtbook}, is just a special case where there is no firm
heterogeneity.

There are several reasons to care about this result. First, it formalizes
Marshall's original motivation of a representative firm and of the
industry supply function in the presence of heterogeneous firms. Second,
it provides a connection between the early literature on LRCE and
the modern literature on industry dynamics (to be reviewed below).
Third, it makes the model of LRCE with heterogeneous firms accessible
to a larger audience (in particular, the example in Section \ref{sec:example}
conveys much of the intuition and can be taught in introductory courses).
Finally, it permits a reinterpretation and extension of the classic
textbook model as a reduced-form representation of a richer economy
with heterogeneous firms and entry and exit.

Our paper links the classic theory of LRCE, which does not explicitly
model dynamics, with the modern literature on \emph{competitive} industry
dynamics started by \citet{lucas1967adjustment}. \citet{lucas1971investment}
developed the first theory of dynamic competitive equilibrium with
stochastic demand, costly capital stock adjustments, and correct (i.e.,
``rational'') expectations about future prices, but firms are homogeneous
and there is no entry and exit. \citet{lucas1978size} studied a model
were firms are heterogeneous, but there is still no entry and exit.
Subsequent developments incorporated both firm heterogeneity and entry
and exit, at the expense of no longer studying the dynamics of capital
accumulation. \citet{jovanovic1982selection} developed the first
of such models. Each period, a firm draws a productivity shock from
a distribution that depends on an unknown productivity type. Firms
have different productivity types and, as they learn their own type,
more productive firms stay and less productive firms exit. The objective
of these papers was to study the dynamic evolution of a competitive
industry, not the steady state. Consequently, all of the interesting
action happens outside the steady state and, indeed, there is no entry
and exit in the steady state of these models.

\citet{hopenhayn1992entry} considers a model with both heterogeneous
firms and entry and exit in the steady state. In contrast to Jovanovic's
model, firms know their productivity types, but productivity types
evolve randomly in such a way that firms that have a low productivity
today can have a high productivity tomorrow and vice versa. As mentioned
earlier, this is the model that we will use to formalize Marshall's
idea that the LRCE of a competitive industry is characterized by the
cost function of a representative firm.\footnote{Hopenhayn's model has been extended in different directions. For example,
\citet{melitz2003impact} adapts Hopenhayn's model to a monopolistically
competitive industry and \citet{10.1257/mac.20150017} incorporate
capital accumulation to Hopenhayn's competitive framework. There is
also a large literature, beginning with the work of \citet{ericson1995markov},
that studies dynamic equilibrium with capital accumulation, stochastic
shocks, heterogeneous firms, and entry and exit under \emph{imperfect}
competition.}

To summarize, previous literature has followed one of two approaches.
In the first approach, all firms are homogeneous, so the existence
of a representative firm follows trivially, thus rendering a convenient
and tractable framework. In the second approach, firm heterogeneity
is explicitly introduced, but the notion of a representative firm
is ignored. Our approach incorporates the best of both worlds: We
explicitly introduce firm heterogeneity and show that, under certain
conditions, a representative firm does exist.\footnote{\citet{chetty1986dynamic} consider a fairly different model of industry
dynamics where firms are heterogeneous and there is entry and exit.
They derive a long-run supply function but show that it cannot be
rationalized as coming from a representative firm.}

For brevity, we focus on the case where input prices are fixed, which
implies that the long-run industry supply function is horizontal.
The extension to the case of input prices that increase with aggregate
quantity was controversial in the early literature; see \citet{opocher2008industry}
for an insightful historical account. The initial approach, by \citet{pigou1928analysis},
\citet{viner1953cost}{[}1931{]}, and others, considered a cost function
that depends both on individual and aggregate quantity. Subsequent
literature (e.g., \citet{kaldor1934equilibrium}, \citet{allen1938mathematical},
and \citet{john1946value})) criticized this reduced-form approach
because of lack of microfoundations. For either approach, the extension
of our result is straightforward: A given aggregate quantity leads
to a given equilibrium input price and, fixing this input price, the
LRCE price is still the minimum point on a representative average
cost function. The long-run industry (inverse) supply function is
simply the mapping from aggregate quantities to these minimum points.
In particular, the aggregate supply function may be increasing if
input prices increase with aggregate quantity.

\section{\label{sec:model}Model and illustrative example}

\subsection{Setup}

We essentially adopt Hopenhayn's (1992) infinite-horizon model of
a competitive industry with a continuum of potential firms, each of
which can produce a homogenous product at total cost $C(q,\theta)$
where $q$ is the quantity produced, $\theta\in\Theta=[\theta_{L},\theta_{H}]\subset\mathbb{R}$
is the firm's type, and $\theta_{L}<\theta_{H}$.

Each period $t=1,2,...$, product demand is given by $Q^{d}(p)$,
where $p\geq0$ is the output price. Firms take price as given and
choose quantity to maximize profit. There is also an infinite mass
of potential entrants with discount factor $\delta\in[0,1]$ who can
decide to enter the market and become a firm. A potential entrant
does not know her type, but knows that her type is independently distributed
according to the probability measure $\nu\in\Delta(\Theta)$. A firm
entering the market pays a one-time entry cost of $\kappa\geq0$.
After paying this cost, a firm immediately learns its own type. Thereafter,
types evolve independently across firms according to the probability
measure $F(\cdot\mid\theta)\in\Delta(\Theta)$, where $\theta$ is
the current type. At the end of the period, each firm makes an exit
decision knowing their current, but not future, type.\footnote{An alternative timing is one where the firm first observes their future
type and then makes exit decisions. All results would go through in
this case, except of course the value function and the corresponding
operators need to be modified accordingly.} At the beginning of each period, a firm learns its type $\theta$
and decides how many units $q$ to produce at a cost $C(q,\theta)$,
where $C(0,\theta)$ is a fixed cost that is sunk once a firm decides
to enter or stay in the industry. There is also an exogenous exit
probability $\rho$. A firm that exits the market (endogenously or
exogenously) does so permanently and obtains a payoff of zero.

We maintain the following assumptions.

\bigskip{}

\begin{assumption}
\label{assu:QdW}Demand: There exists a choke price $v>0$ (possibly
infinity) such that $Q^{d}(\cdot)$ is continuous and decreasing for
all $p\in(0,v)$, and $Q^{d}(p)=0$ for all $p\geq v$.
\end{assumption}
\smallskip{}

\begin{assumption}
\label{assu:costs}Costs: For all $\theta\in\Theta$: $C(\cdot,\theta)$
is continuously differentiable, with $C(q,\theta)\geq0$, $C'(q,\theta)\geq0$,
and $C''(q,\theta)>0$ for all $q\geq0$, and $\lim_{q\rightarrow\infty}C'(q,\theta)=\infty$;
For all $q\geq0$, $C(q,\cdot)$ is increasing.
\end{assumption}
\smallskip{}

\begin{assumption}
\label{assu:FOSD}Order over types: For any $\theta_{1}<\theta_{2}$,
$F(\theta\mid\theta_{2})<F(\theta\mid\theta_{1})$ for all $\theta\in(\theta_{L},\theta_{H})$.
\end{assumption}
\smallskip{}

\begin{assumption}
\label{assu:rho>0} The exogenous probability of exit is positive,
i.e., $\rho>0$.
\end{assumption}
\smallskip{}

\begin{assumption}
\label{assu:CDFs_cont}Measures over types: (i) $\nu$ has a continuous
probability density function (pdf), $f_{\nu}(\cdot)$, with support
equal to $\Theta$; (ii) For all $\theta$: $F(\cdot\mid\theta)$
has a pdf $f(\cdot\mid\theta)$, with support equal to $\Theta$,
and $(\theta',\theta)\mapsto f(\theta'\mid\theta)$ is jointly continuous.
\end{assumption}
\bigskip{}

Assumption \ref{assu:QdW} implies the existence of a downward sloping
inverse demand function, $P^{d}(\cdot)$. Assumption \ref{assu:costs}
implies existence and uniqueness of an optimal quantity
\[
q(p,\theta)\equiv\arg\max_{q\geq0}pq-C(q,\theta)
\]
The assumption also implies that the profit function
\[
\pi(p,\theta)\equiv pq(p,\theta)-C(q(p,\theta),\theta).
\]
is nonincreasing in $\theta$, and decreasing for $(p,\theta)$ such
that $q(p,\theta)>0$.

Assumption \ref{assu:FOSD} postulates a first-order stochastic dominance
relationship across types, so that higher types today are more likely
to become higher types tomorrow. Assumption \ref{assu:rho>0} guarantees
that the life span of a firm is almost surely finite; in particular,
if there is no entry, then there must be zero aggregate production
in equilibrium. This assumption is made for simplicity, puts the focus
on equilibria with positive entry, and allows us to include the special
case where firms' types are permanent. Here, we differ from \citet{hopenhayn1992entry},
who instead assumes that $\rho=0$, guarantees finite lifespan with
an additional recurrence condition on $F$ that rules out permanent
types, and subsequently restricts attention to equilibria with positive
entry. Our assumption is an alternative way to capture the same phenomenon
and we prefer it because it allows us to study the simple case where
types are permanent (see below for an example).

Finally, Assumption \ref{assu:CDFs_cont} lists technical conditions
regarding the measures over types.

\smallskip{}

\begin{assumption}
\label{assu:prod>0}$\pi(v,\theta_{H})>\kappa$.
\end{assumption}
\smallskip{}

Assumption \ref{assu:prod>0} is made for simplicity. It rules out
equilibria with zero aggregate production by requiring that even the
highest-cost firm prefers to enter whenever price equals the maximum
willingness to pay, $v$.

The expected net present discounted value of a firm of type $\theta$
who faces (steady-state) price $p$ every period is
\begin{equation}
V(p,\theta)=\pi(p,\theta)+\delta(1-\rho)\max\left\{ \int_{\Theta}V(p,\theta')F(d\theta'\mid\theta),0\right\} .\label{eq:Bellman}
\end{equation}

Assumption \ref{assu:FOSD} and the fact that $\pi(p,\cdot)$ is decreasing
imply that $\int_{\Theta}V(p,\theta')F(d\theta'\mid.)$ is decreasing.
Therefore, the optimal exit decision in steady state is characterized
by a marginal type $m\in\Theta$ with the property that all lower
types stay and all higher types exit the market.

Let $\mu(n,m)$ denote the steady-state measure of types of firms
given the mass of entrants $n\geq0$ and the marginal type $m\in\Theta$.
In particular, for any Borel set $A\subseteq\Theta$,
\begin{equation}
\mu(n,m)(A)=\nu(A)n+(1-\rho)\int_{\theta_{L}}^{m}F(A\mid\theta)\mu(n,m)(d\theta).\label{eq:defMU}
\end{equation}
The assumption that $\rho>0$ guarantees existence of a steady-state
measure. 

The corresponding aggregate supply at price $p$ is
\[
Q^{s}(p;n,m)\equiv\int_{\Theta}q(p,\theta)\mu(n,m)(d\theta).
\]

\bigskip{}

\begin{defn}
\label{def:LRCE}A tuple $\left\langle p^{e},n^{e},m^{e}\right\rangle $
is a \textbf{long-run competitive equilibrium (LRCE)} if the following
conditions are satisfied:

(i) Profit maximization and market clearing: $Q^{d}(p^{e})=Q^{s}(p^{e};n^{e},m^{e})$.

(ii) Free entry: $\int_{\Theta}V(p^{e},\theta)\nu(d\theta)\leq\kappa$,
with equality if $n^{e}>0$.

(iii) Optimal exit: $\int_{\Theta}V(p^{e},\theta')F(d\theta'\mid m^{e})=0$
if $m^{e}\in(\theta_{L},\theta_{H})$, $\geq0$ if $m^{e}=\theta_{H}$,
and $\leq0$ if $m^{e}=\theta_{L}$.

\end{defn}
\bigskip{}

An LRCE captures the steady state of the dynamic competitive industry.\footnote{\citet{hopenhayn1992entry} called an LRCE a stationary equilibrium
and showed that it corresponds to the steady state of a perfect foresight
equilibrium of the dynamic environment.} The first condition requires market clearing and already incorporates
the assumption of profit maximization. The second condition is a free
entry condition that requires the net present value of entry to equal
the entry cost if the mass of entrants is positive. The third condition
requires the marginal type to be indifferent between staying or exiting
the market, provided it is an interior type.

\bigskip{}

\begin{lem}
\label{lem:EqmQ>0}In any LRCE, both aggregate production and entry
must be positive.
\end{lem}
\begin{proof}
Suppose $p^{e}$ is an LRCE price and $Q^{d}(p^{e})=0$. By Assumption
\ref{assu:QdW}, $p^{e}\geq v$. By the fact that $\pi(p,\theta)$
is nondecreasing in $p$ and nonincreasing in $\theta$ and by Assumption
\ref{assu:prod>0}, $\pi(p^{e},\theta)\geq\pi(v,\theta)\geq\pi(v,\theta_{H})>\kappa$
for all $\theta$. Thus, $V(p^{e},\theta)>\kappa$ for all $\theta$,
so that $p^{e}$ does not satisfy the free entry condition (ii) in
Definition \ref{def:LRCE}, contradicting the fact that $p^{e}$ is
an LRCE price. Therefore, $Q^{d}(p^{e})>0$ and by condition (i) in
Definition \ref{def:LRCE}, $Q^{s}(p^{e};n^{e},m^{e})>0$, which then
implies, by the assumption that $\rho>0$, that $n^{e}>0$.
\end{proof}
\bigskip{}

\begin{defn}
\label{def:supply}The \textbf{long-run industry (inverse) supply
function} is a function $Q\mapsto P_{LR}^{s}(Q)$ with the property
that, for any $Q>0$, $p=P_{LR}^{s}(Q)$ is the unique price satisfying
the following conditions for some $m\in\Theta$ and $n>0$:

(i) $Q=Q^{s}(p;n,m)$.

(ii) $\int_{\Theta}V(p,\theta)\nu(d\theta)=\kappa$.

(iii) $\int_{\Theta}V(p,\theta')F(d\theta'\mid m)=0$ if $m\in(\theta_{L},\theta_{H})$,
$\geq0$ if $m=\theta_{H}$, and $\leq0$ if $m=\theta_{L}$.
\end{defn}
\bigskip{}

The next result follows immediately from the definitions and from
Lemma \ref{lem:EqmQ>0}'s implication that the free entry condition
in Definition \ref{def:supply} holds with equality in equilibrium.
\begin{prop}
\label{prop:D=00003DS}Suppose that the long-run industry supply function
$P_{LR}^{s}(\cdot)$ exists. Then $p^{e}$ is part of an LRCE if and
only if $p^{e}=P_{LR}^{s}(Q^{d}(p^{e}))$ and $Q^{d}(p^{e})>0$.
\end{prop}
\bigskip{}

Proposition \ref{prop:D=00003DS} simply says that the LRCE price
is such that supply equals demand. When firms are identical, it is
well known that the long-run industry supply function is horizontal
at the minimum point of the average cost function. Our objective is
to characterize this function for the environment described in this
section, where firms are heterogeneous.

\subsection{\label{sec:example}A simple example}

We discuss an example that is simple enough to be taught in introductory
courses and conveys much (but not all) of the intuition behind our
results. We assume that: (i) there are only two types, not a continuum,
$\theta_{H}>\theta_{L}\geq0$, and each type is equally likely to
be drawn by an entrant; (ii) $C(q,\theta)=c(q)+\theta$, so that a
firm's type represents its fixed cost and all firms have the same
marginal cost $MC(q)\equiv c'(q)$; (iii) the entry cost is zero,
$\kappa=0$; (iv) types are permanent, so that a firm keeps the type
it draws upon entry for its entire lifetime; and (v) firms are impatient,
$\delta<1$. The variable cost function $c(\cdot)$ satisfies the
following conditions: $c(0)=0$, $c'(0)=0$, $c'(q)>0$ and $c''(q)>0$
for all $q>0$, and $\lim_{q\rightarrow\infty}c'(q)=\infty$.

$\textsc{steady-state measure of types.}$ It is easy to see that
type $\theta_{L}$ will stay and type $\theta_{H}$ will exit in equilibrium;
in particular, we will drop $m$ from the notation.\footnote{For the free entry condition to hold for $p>0$, the profit of type
$\theta_{H}$ must be negative. Because types are permanent, type
$\theta_{H}$ will find it optimal to exit.} The steady-state mass of firms of type $\theta_{L}$, denoted by
$\mu_{L}$, is determined by the steady-state mass of entrants, $n$,
as follows:
\begin{equation}
\mu_{L}=n/2+\mu_{L}(1-\rho).\label{eq:ststL}
\end{equation}
The RHS of equation (\ref{eq:ststL}) is the sum of the mass of entrants
of type $\theta_{L}$, $n/2$, and the mass of firms of type $\theta_{L}$
that were already present and did not exit exogenously, $\mu_{L}(1-\rho)$.
The equation implies that, in steady state, the mass of type $\theta_{L}$
remains constant. For firms of type $\theta_{H}$, who never stay
for more than one period, their mass is half the mass of entrants.
Thus, the steady-state masses of firms of each type as a function
of the mass of entrants, $n$, are 
\[
\mu_{L}(n)=n/(2\rho)\,\,\,\,\,\mbox{and}\,\,\,\,\,\mu_{H}(n)=n/2.
\]

$\textsc{long-run industry supply function}.$ The conditions in the
definition of the long-run supply function become:\smallskip{}

(i) $Q=(\mu_{L}(n)+\mu_{H}(n))q(p)>0$.

(ii) (Free entry) $NPV(p)\equiv\frac{1}{2}\pi(p,\theta_{L})/(1-\delta(1-\rho))+\frac{1}{2}\pi(p,\theta_{H})=0$.

\smallskip{}

Condition (i) requires aggregate output supply to equal $Q$. Condition
(ii) requires that the net present value of an entrant is zero. With
probability $1/2$, a firm is of type $\theta_{L}$ and remains in
the market until it has to exogenously exit, thus expecting a net
present value of $\pi(p,\theta_{L})/(1-\delta(1-\rho))$. With probability
$1/2$, a firm is of type $\theta_{H}$, makes profit $\pi(p,\theta_{H})$,
and exits the market.

The weights on the profit functions of each type in the free-entry
condition have an intuitive interpretation. The weight $\Lambda_{L}\equiv1/(2(1-\delta(1-\rho))$
on $\pi(p,\theta_{L})$ is equal to the steady-state mass of type
$\theta_{L}$, normalized by the mass of entrants $n$, in a hypothetical
world where firms, instead of exiting with probability $\rho$, exit
with probability $1-\delta(1-\rho)$.\footnote{Formally, $\Lambda_{L}\equiv\mu_{E}(n,\delta)(\theta_{L})/n$, where
$\mu_{E}(n,\delta)(\theta_{L})$ solves $\mu_{E}(n,\delta)(\theta_{L})=n/2+\mu_{E}(n,\delta)(\theta_{L})\delta(1-\rho)$.} The hypothetical and actual probabilities of exit coincide as $\delta\rightarrow1$,
and so the weight asymptotically equals the actual, normalized steady-state
mass of type $\theta_{L}$. Similarly, the weight $\Lambda_{H}\equiv1/2$
on $\pi(p,\theta_{H})$ is equal to the normalized steady-state mass
of firms of type $\theta_{H}$ (here, $\delta$ is irrelevant because
type $\theta_{H}$ exits with probability 1). Thus, the net present
value of entry can be written as 
\begin{align}
NPV(p) & =\Lambda_{L}\pi(p,\theta_{L})+\Lambda_{H}\pi(p,\theta_{H})\nonumber \\
 & =pq(p)(\Lambda_{L}+\Lambda_{H})-(\Lambda_{L}C(q(p),\theta_{L})+\Lambda_{H}C(q(p),\theta_{H})).\label{eq:NPV_new}
\end{align}

By equation (\ref{eq:NPV_new}), the solution $p^{e}$ to $NPV(p^{e})=0$
satisfies 
\begin{equation}
p^{e}=AC^{e}(q(p^{e}),\Lambda)\equiv\frac{\Lambda_{L}AC(q(p^{e}),\theta_{L})+\Lambda_{H}AC(q(p^{e}),\theta_{H})}{(\Lambda_{L}+\Lambda_{H})},\label{eq:p=00003DAC^e}
\end{equation}
where $\Lambda\equiv(\Lambda_{L},\Lambda_{H})$, $AC(q,\theta)\equiv C(q,\theta)/q$
is the average cost of type $\theta$, and $q\mapsto AC^{e}(q,\Lambda)$
is a weighted average cost function.

By profit maximization, $p^{e}=MC(q(p^{e}))$, and so (\ref{eq:p=00003DAC^e})
implies that $p^{e}$ equalizes marginal and weighted average cost,
\begin{equation}
p^{e}=MC(q(p^{e}))=AC^{e}(q(p^{e}),\Lambda).\label{eq:p=00003DAC(p)}
\end{equation}

The left panel of Figure \ref{fig:exampleLRCE} illustrates how to
find $p^{e}$. The figure plots the marginal cost function common
to all types, $MC(\cdot)$, the average cost function for each type,
$AC(\cdot,\theta)$, and the weighted average cost function $AC^{e}(\cdot,\Lambda)$.
The zero-profit price $p^{e}$ is given by the intersection of the
marginal cost and weighted average cost functions, and this intersection
occurs at the minimum point on the weighted average cost function.\footnote{For a proof that the intersection occurs at the minimum point of $AC^{e}(\cdot,\Lambda)$,
note that the first order condition for the problem $\min_{q}AC^{e}(q,\Lambda)$
is precisely the condition $MC(q)=AC^{e}(q,\Lambda)$. Moreover, the
second order condition is satisfied because $c''(q)>0$ for all $q>0$.} Therefore, $q(p^{e})=q_{min}^{e}\equiv\arg\min_{q}\,\,AC^{e}(q,\Lambda)$
and the zero-profit price $p^{e}$ is
\[
p^{e}=AC^{e}(q_{min}^{e},\Lambda)=\min_{q}AC^{e}(q,\Lambda).
\]

\begin{figure}
\begin{centering}
	\newpsobject{showgrid}{psgrid}{subgriddiv=1,griddots=10,gridlabels=7pt,gridcolor=red} 
		\psset{unit=6mm} 	
	\psset{algebraic=true} 	
	\begin{pspicture}(0,0)(23,10) 

	\psline[linewidth=0.05]{->}(0,0)(10,0) 	
	\psline[linewidth=0.05]{->}(0,0)(0,10) 	
\rput(9,-0.5){quantity} 	
	\rput(-1.2,9){price \&} 	
	\rput(-1.1,8.5){costs}	 	
\psframe*[linecolor=pink, 		fillcolor=pink, 		fillstyle=vlines](0,4.70)(5.40,3.40) 	
	\rput(3.2,4.3){\red{\small{ $\pi^{e}>0$}}} 	
	
\psline(0,0)(8,7) 		
\rput(8,7.5){{ $MC$}} 	
	\pscurve[linestyle=dashed,linecolor=gray](0.2,2.5)(0.7,2)(5.7,2.5) 	
	\rput(6,2){{\gray{ $AC(\cdot,\theta_{L})$   }}}    	
		\pscurve(0.5,7.5)(2.5,3.5)(7.5,4) 	
	\rput(8,3.5){{{ $AC^{\ast}$   }}}    	
	
		\pscurve(2,9)(4,5)(9,5.5) 
		\rput(9.5,5){{{ $AC^{e}(\cdot,\Lambda)$   }}}    	
	
		\pscurve[linestyle=dashed,linecolor=gray](3.8,10.5)(5.8,6.5)(10.8,7) 		
\rput(4.5,11){{\gray{ $AC(\cdot,\theta_{H})$   }}}      	
	
	\psline[linestyle=dotted](0,3.20)(17.4,3.20)		 	
	\psline[linestyle=dotted](0,4.7)(15.5,4.7) 	
	\rput(-0.5,4.7){{ $p^{e}$}}		 	
		\psline[linestyle=dotted](5.4,0)(5.40,4.7)		
	 		\psline[linestyle=dotted](3.7,0)(3.7,3.2)		 	
	
		\rput(3.60,-0.5){{ $q^{\ast}_{min}$}}  	
	\rput(5.50,-0.5){{ $q^{e}_{min}$}}

\psline[linewidth=0.05]{->}(12,0)(22,0) 	
	\psline[linewidth=0.05]{->}(12,0)(12,10) 	
		\rput(11.2,9){price} 	
	
	\pscurve(14,8.1)(16,4.1)(21,2.1) 	
	\rput(14,8.5){$P^{d}$}  
	
\psline(12,4.7)(21,4.7) 	
	\rput(21.6,5){$S_{LR}$}

\psline[linestyle=dotted](15.43,0)(15.43,4.70)  	
	\psline[linestyle=dotted](17.43,0)(17.43,3.20) 
	
\rput(15.43,-0.5){$Q^{e}$}  	
	\rput(17.43,-0.5){$Q^{\ast}$}  	
	\rput(22.1,-0.5){aggregate}  	
	\rput(22.1,-1.0){~quantity}  	
	
		\end{pspicture}
\par\end{centering}
\bigskip{}

\bigskip{}

\caption{\label{fig:exampleLRCE}Long-run competitive equilibrium in the example.\protect \\
\emph{\footnotesize{}The left panel shows the marginal cost curve
(which is the same for all firms) and the average cost curve for high
and low types. It also shows the average of the average cost curves,
$AC^{e}$, when the average is taken with respect to the ex-ante measure
of types as perceived by an entrant who discounts the future. The
equilibrium price $p^{e}$ is given by the minimum of the $AC^{e}$
curve, where potential entrants make zero profits and produce $q_{min}^{e}$.
The left panel also shows the average of the average cost curves,
$AC^{*}$, when the average is taken with respect to the actual steady-state
measure of firms. Selection in exit implies that low-cost firms stay
and high-cost firms exit, and so $AC^{*}$ lies below $AC^{e}$. In
particular, actual firms operating in the steady-state make positive
profits, as shown by the rectangular area $\pi^{e}>0$. The right
panel plots the industry demand and the long-run supply function $S_{LR}$,
which is a horizontal line at price $p^{e}$. The equilibrium aggregate
quantity, $Q^{e}$, can be found by evaluating the demand function
at the equilibrium price $p^{e}$. A planner that maximizes steady-state
surplus would instead choose individual quantity $q_{min}^{*}$ that
minimizes $AC^{*}$, more entry, and a higher aggregate quantity,
$Q^{*}$.}}
\end{figure}

Finally, it is straightforward to check that, since $p^{e}>0$, there
exists $n(Q)>0$ satisfying condition (i) in Definition \ref{def:supply},
i.e., $Q=(\mu_{L}(n(Q))+\mu_{H}(n(Q)))q(p^{e})$. Therefore, the
long-run supply function exists and is horizontal at the price that
minimizes the weighted average cost function $AC^{e}(\cdot)$. Thus,
provided that $P^{d}(0)>p^{e}$, there exists a unique LRCE where
price is $p^{e}$ and the mass of entrants $n^{e}$ is such that the
product market clears, i.e., $Q^{d}(p^{e})=(\mu_{L}(n^{e})+\mu_{H}(n^{e}))q_{min}^{e}$.\footnote{The solution is unique and given by $n^{e}=Q^{d}(p^{e})/((1/2\rho+1/2)q_{min}^{e})$.}

Figure \ref{fig:exampleLRCE} also illustrates that aggregate profits
are strictly positive in an LRCE. The equilibrium profit of the average
firm is $\pi^{e}\equiv(p^{e}-AC^{*}(q_{min}^{e}))q_{min}^{e}>0$,
where $q_{min}^{e}$ is the quantity produced by each firm and

\[
AC^{*}(\cdot)\equiv\frac{(1/(2\rho))AC(\cdot,\theta_{L})+(1/2)AC(\cdot,\theta_{H})}{((1/(2\rho))+1/2)}
\]
is the per-unit cost function of the average firm producing in equilibrium.
The weights in $AC^{*}(\cdot)$ correspond to the steady-state proportion
of firms of each type. While these weights converge to $\Lambda$
as $\delta\rightarrow1$, for the case $\delta<1$, $AC^{*}(\cdot)$
puts more weight on the low cost type relative to $AC^{e}(\cdot,\Lambda)$.
Intuitively, the selection in exit implies that the steady-state composition
of firms is tilted towards low-cost firms relative to the ex-ante
perception of a potential entrant who discounts the future. Thus,
while potential entrants make zero profits ex-ante, the actual firms
operating in the steady state make strictly positive profits.

Consequently, a planner who wishes to maximize steady-state surplus
prefers a higher aggregate quantity $Q^{*}$, a lower quantity per
firm $q_{min}^{*}$, and a higher mass of entrants $n^{*}$ compared
to the LRCE quantities $Q^{e}$, $q_{min}^{e}$, and $n^{e}$. The
planner's preferred outcome is not an equilibrium outcome, because
the net present value of entry would be negative and firms would not
enter to begin with. Of course, a planner may not want to maximize
steady-state surplus, preferring alternatively to take the entire
path into account. While the equilibrium does not maximize the steady-state
value of the total surplus, it does maximize the net present value
of total surplus in \citet{hopenhayn1992entry}'s dynamic model. But
this example highlights that, when firms are heterogeneous, one has
to be explicit about the planner's objective.\footnote{This point is analogous to the result that the Ramsey model does not
deliver the golden-rule rate of saving. See \citet{atkeson2005modeling}
for a similar point in the context of a model used to study the life
cycle of manufacturing plants.}

In the special case where there is a single type, $\theta_{L}=\theta_{H}$,
the standard textbook results hold: The industry supply function is
horizontal at the price that equals the minimum of the average cost
function (all firms have the same cost function), each firm makes
zero profits, and aggregate surplus is maximized in an LRCE (irrespective
of the value of the discount factor $\delta$). Alternatively, we
can interpret the standard textbook model as a case where firms are
of different types but know their types before entering the market.
In that case, only firms of type $\theta_{L}$ will operate in the
market in an LRCE.

$\textsc{beyond the simple example.}$ We extend the logic in the
example in several directions. First, marginal costs may differ by
type. We will tackle this case by expressing the average cost function
in terms of price, not quantity. To anticipate how things would change,
consider the more general case where the optimal quantity may differ
by type due to different marginal costs. Equation (\ref{eq:p=00003DAC^e})
becomes 
\begin{equation}
p^{e}=\frac{\Lambda_{L}C(q(p^{e},\theta_{L}),\theta_{L})+\Lambda_{H}C(q(p^{e},\theta_{H}),\theta_{H})}{\Lambda_{L}q(p^{e},\theta_{L})+\Lambda_{L}q(p^{e},\theta_{H})}.\label{eq:diffMC}
\end{equation}
Since $q(p^{e},\theta_{L})\neq q(p^{e},\theta_{H})$, it follows that
the RHS of (\ref{eq:diffMC}) is no longer the average of the average
costs, as it was in expression (\ref{eq:p=00003DAC^e}). But one could
still think of it as a type of average cost function where the argument
is price and not quantity. To see this point, note that the numerator
of the RHS is a weighted cost function; denote it by $\bar{C}$. Also,
the denominator is a weighted quantity; denote it by $\bar{q}$. Letting
$\bar{AC}=\bar{C}/\bar{q}$ denote the average weighted cost function,
it follows that (\ref{eq:diffMC}) can be rewritten as $p^{e}=\bar{AC}(p^{e})$.
Together with the fact that firms choose quantities to equate price
to marginal cost, one can use simple algebra to show that the equilibrium
price $p^{e}$ minimizes the $\bar{AC}$ curve.\footnote{This is true because the condition for minimization of $\bar{AC}$
is $\bar{C}'(p)/\bar{q}'(p)=\bar{AC}(p)$ and the fact that $p=dc(q,\theta)/dq$
for all types implies that $\bar{C}'(p)/\bar{q}'(p)=p$.} In the next section, we will show that this is true in more general
cases.

The second extension we tackle is that types may be non-permanent.
When types follow a more general Markov process, optimal entry decisions
are the solution to a non-trivial dynamic optimization problem. We
will use results from the theory of bounded linear operators to show
that, nevertheless, ex-ante expected profits can still be expressed
as the weighted average of the profits of each type. Third, there
may be a continuum of non-permanent types. In this case, exit decisions
also need to be characterized by solving a dynamic optimization problem.
Fourth, strictly positive entry costs need to be incorporated into
the definition of average cost.

\section{\label{sec:charact}Characterization of long-run industry supply}

To state the main result, we first define an average weighted cost
function. Letting $\mathcal{M}(\Theta)$ be the space of finite Borel
measures that are absolutely continuous with respect to Lebesgue,
we define $\bar{C}:[0,\infty)\times\mathcal{M}(\Theta)\rightarrow[0,\infty)$
as
\[
\bar{C}(p,\eta)=\int_{\Theta}C(q(p,\theta),\theta)\eta(d\theta)+\kappa
\]
for all $p\geq0$ and $\eta\in\mathcal{M}(\Theta)$. This is the weighted
cost with respect to a measure $\eta$. Similarly, let $\bar{q}:[0,\infty)\times\mathcal{M}(\Theta)\rightarrow[0,\infty)$
be defined by
\[
\bar{q}(p,\eta)=\int_{\Theta}q(p,\theta)\eta(d\theta).
\]

The corresponding \emph{average weighted cost} function is then defined
by
\[
\bar{AC}(p,\eta)\equiv\bar{C}(p,\eta)/\bar{q}(p,\eta),
\]
provided that $\bar{q}(p,\eta)>0$.\footnote{If $\bar{q}(p,\eta)=0$, we define $\bar{AC}(p,\eta)=\infty$. }
In the case where marginal costs are identical, the average weighted
cost coincides with the weighted average cost, as in the example,
but this is not true in general.

Next, for each $n$, $m$, and $\delta$, we define $\mu_{E}(n,m,\delta)\in\mathcal{M}(\Theta)$
to be the steady-state measure of types of firms when the mass of
entrants is $n$, firms survive with exogenous probability $\delta(1-\rho)$,
and surviving firms exit endogenously if their type is lower than
$m\in\Theta$, i.e., for any Borel set $A\subseteq\Theta$,
\[
\mu_{E}(n,m,\delta)(A)=\nu(A)n+\delta(1-\rho)\int_{\theta_{L}}^{m}F(A\mid\theta)\mu_{E}(n,m,\delta)(d\theta).
\]
For the special case of $\delta=1$, $\mu_{E}(n,m,1)=\mu(n,m)$ is
the \emph{actual} steady-state measure of types defined in equation
(\ref{eq:defMU}), because in the model firms survive with exogenous
probability $1-\rho$, not $\delta(1-\rho)$.

Finally, since $\mu_{E}$ is linear in $n$, we define the normalized
mass
\[
\Lambda(m,\delta)\equiv\mu_{E}(n,m,\delta)/n\in\mathcal{M}(\Theta).
\]

We now state the main result.

\bigskip{}

\begin{thm}
\label{thm:main_charact}The long-run industry supply function exists
and it is given by 
\[
P_{LR}^{S}(Q)=\min_{p,m}\bar{AC}(p,\Lambda(m,\delta))
\]
for any $Q>0$.
\end{thm}
\bigskip{}

Theorem \ref{thm:main_charact} extends the textbook characterization
of the long-run supply function to a setting with heterogeneous firms.
The long-run supply function is horizontal at a price that minimizes
the average weighted cost function, where the minimum is with respect
to both price and the marginal type. The average weight cost function
is constructed using the measure $\Lambda(m,\delta)$, which can be
viewed as the normalized steady-state cross-sectional distribution
of firm types in a hypothetical world where firms survive with exogenous
probability $\delta(1-\rho)$ and surviving firms exit endogenously
if their type is lower than $m$.

In particular, Theorem \ref{thm:main_charact} formalizes Marshall's
notion of a representative firm as a hypothetical firm with average
cost function $\bar{AC}$. In the special case where all firms have
identical marginal cost functions (as in the example), the average
cost function of the representative firm, $\bar{AC}$, corresponds
to a weighted average of the average cost functions.

\bigskip{}

\begin{cor}
\label{cor:uniqueness}There exists a unique LRCE and it is characterized
by positive entry and positive aggregate production.
\end{cor}
\begin{proof}
Follows immediately from Proposition \ref{prop:D=00003DS}, Theorem
\ref{thm:main_charact}, and the fact that assumption \ref{assu:prod>0}
and monotonicity of $\pi(\cdot,\theta)$ imply that $\min_{p,m}\bar{AC}(p,\Lambda(m,\delta))<v$.
\end{proof}

\subsection{\label{sec:proof}Proof of Theorem \ref{thm:main_charact}}

We will show that there is a unique solution $(p^{e},m^{e})$ to equations

(ii) $\int_{\Theta}V(p,\theta)\nu(d\theta)=\kappa$, and

(iii) $\int_{\Theta}V(p,\theta')F(d\theta'\mid m)=0$ if $m\in(\theta_{L},\theta_{H})$,
$\geq0$ if $m=\theta_{H}$, and $\leq0$ if $m=\theta_{L}$

\hspace{-.65cm}in Definition \ref{def:supply}, and that this solution
satisfies
\[
(p^{e},m^{e})=\min_{p,m}\bar{AC}(p,\Lambda(m,\delta)).
\]
The proof has three steps. Throughout the proof, we let $\varrho\equiv\delta(1-\rho)$.

\bigskip{}

$\textsc{step 1}.$ For any $(p,\theta)\in\mathbb{R}_{+}\times\Theta$
and $m\in\Theta$, let 
\begin{equation}
V_{m}(p,\theta)=\pi(p,\theta)+\varrho T_{m}[V_{m}(p,\cdot)](\theta)\label{eq:V_m}
\end{equation}
where $T_{m}[g](\theta)=1\{\theta\leq m\}\int_{\Theta}g(\theta')F(d\theta'\mid\theta)$.
In words, $V_{m}$ differs from the value function $V$ defined in
equation (\ref{eq:Bellman}) in that it forces a possibly suboptimal
exit decision threshold $m$.

Consider the system of equations:\bigskip{}

(ii') $\int_{\Theta}V_{m}(p,\theta)\nu(d\theta)=\kappa$, and

(iii') $\int_{\Theta}V_{m}(p,\theta')F(d\theta'\mid m)=0$ if $m\in(\theta_{L},\theta_{H})$,
$\geq0$ if $m=\theta_{H}$, and $\leq0$ if $m=\theta_{L}$.

\bigskip{}

We will show that we can work with the system of equations (ii')-(iii')
rather than (ii)-(iii).
\begin{lem}
If $(p,m)$ is the unique solution to (ii')-(iii'), then $(p,m)$
must also be the unique solution to (ii)-(iii).
\end{lem}
\begin{proof}
Let $(p,m)$ be the unique solution to (ii')-(iii'). In particular,
$m$ is the unique solution to (iii') given $p$. Let $m_{0}$ be
the optimal exit threshold given $p$. In particular, $\int_{\Theta}V_{m_{0}}(p,\theta')F(d\theta'\mid m)=0$
if $m\in(\theta_{L},\theta_{H})$, $\geq0$ if $m=\theta_{H}$, and
$\leq0$ if $m=\theta_{L}$. Since $m$ is the unique solution to
(iii') given $p$, it follows that $m=m_{0}$ and, therefore, $V_{m}=V_{m_{0}}$.
In addition, by optimality of $m_{0}$ and the one-shot deviation
principle, $V_{m_{0}}=V$. Therefore, $(p,m)$ solves (ii)-(iii).
To show uniqueness, suppose that $(p',m')$ solves (ii)-(iii). Then
$V=V_{m'}$ and so $(p',m')$ must also solve (ii')-(iii'). But since
$(p,m)$ is the unique solution to (ii')-(iii'), it must be that $(p',m')=(p,m)$.
\end{proof}
\bigskip{}

$\textsc{step 2}.$ In this step, we will show that conditions (ii')-(iii')
can be equivalently expressed using weighted profit functions. This
is one of the main insights of the proof and it relies on the concept
of the adjoint of a bounded operator to identify the appropriate weight
over profit functions.

For each $m\in\Theta$, define an operator $\Phi_{m}:\mathcal{M}(\Theta)\rightarrow\mathcal{M}(\Theta)$
such that, for all $A\subseteq\Theta$ Borel,
\[
\Phi_{m}[\eta](A)=\int_{\theta_{L}}^{m}F(A\mid\tilde{\theta})\eta(d\tilde{\theta}).
\]
$\Phi_{m}[\eta]$ gives the measure of types that results from applying
the Markov operator $F$ to current types that are below the marginal
type $m$, when the measure of current types is $\eta$.

The next result collects two useful properties of the operator $\Phi_{m}$.

\bigskip{}

\begin{lem}
\label{lemma_Phi} (i) For any $\varrho\in[0,1)$ and $m\in\Theta$,
$\sum_{j=0}^{\infty}\varrho^{j}\Phi_{m}^{j}=\left(I-\varrho\Phi_{m}\right)^{-1}$
is a bounded operator from $\mathcal{M}(\Theta)$ to itself, where
$I$ is the identity operator; (ii) For all $j$, $\Phi_{m}^{j}$
is the adjoint operator of $T_{m}^{j}$.
\end{lem}
\begin{proof}
See the Appendix.
\end{proof}
\bigskip{}

Using the operator $\Phi_{m}$, $\mu_{E}$ can be alternatively written
as
\[
\mu{}_{E}(n,m,\delta)=\nu n+\varrho\Phi_{m}[\mu{}_{E}(n,m,\delta)].
\]
Analogously, we can define $\mu{}_{X}(n,m,\delta)\in\mathcal{M}(\Theta)$
as the same measure, except that the distribution of entrants is the
one facing the marginal exit type, $F(\cdot\mid m)$, i.e.,
\[
\mu{}_{X}(n,m,\delta)=F(\cdot\mid m)n+\varrho\Phi_{m}[\mu{}_{X}(n,m,\delta)].
\]

By Lemma \ref{lemma_Phi}(i),
\begin{equation}
\Lambda(m,\delta)=\mu{}_{E}(n,m,\delta)/n=(I-\varrho\Phi_{m})^{-1}[\nu]\label{Lambda}
\end{equation}
and
\[
\Lambda{}_{X}(m,\delta)\equiv\mu{}_{X}(n,m,\delta)/n\equiv(I-\varrho\Phi_{m})^{-1}[F(\cdot\mid m)].
\]

Our goal is to show that we can express the value functions in terms
of weighted profit functions, with weights $\Lambda$ and $\Lambda_{X}$
for the entry and exit conditions, respectively. For this purpose,
we define the weighted profit function $\bar{\pi}:[0,\infty)\times\mathcal{M}(\Theta)\rightarrow\mathbb{R}$,
where
\[
\bar{\pi}(p,\eta)=\int\pi(p,\theta)\eta(d\theta)
\]
for all $p\geq0$ and $\eta\in\mathcal{M}(\Theta)$. We then state
the following two conditions, which the next lemma will show to be
equivalent to conditions (ii')-(iii').

\bigskip{}

Condition (ii''). $\bar{\pi}(p,\Lambda(m,\delta))=\kappa$.

Condition (iii''). $\bar{\pi}(p,\Lambda_{X}(m,\delta))=0$ if $m\in(\theta_{L},\theta_{H})$,
$\geq0$ if $m=\theta_{H}$, and $\leq0$ if $m=\theta_{L}$.\bigskip{}

\begin{lem}
$(p,m)$ solves (ii')-(iii') if and only if it solves (ii'')-(iii'').
\end{lem}
\begin{proof}
By repeatedly applying equation (\ref{eq:V_m}), it follows that 
\[
V_{m}(p,\theta)=\sum_{j=0}^{\infty}\varrho^{j}T_{m}^{j}[\pi(p,\theta)](\theta).
\]
Then 
\begin{align*}
\int V_{m}(p,\theta)\nu(d\theta) & =\int\sum_{j=0}^{\infty}\varrho^{j}T_{m}^{j}[\pi(p,\cdot)](\theta)\nu(d\theta)\\
 & =\int\pi(p,\theta)\left(\sum_{j=0}^{\infty}\varrho^{j}\Phi_{m}^{j}[\nu](d\theta)\right)\\
 & =\int\pi(p,\theta)\left(I-\varrho\Phi_{m}\right)^{-1}[\nu](d\theta)\\
 & =\int\pi(p,\theta)\Lambda(m,\delta)(d\theta)=\bar{\pi}(p,\Lambda(m,\delta)),
\end{align*}
where the second line follows because $\Phi_{m}^{j}$ is the adjoint
operator of $T_{m}^{j}$ (see Lemma \ref{lemma_Phi}(ii)) and the
last line follows by definition of $\Lambda$ in equation (\ref{Lambda}).
A similar argument establishes $\int_{\Theta}V_{m}(p,\theta')F(d\theta'\mid m)=\bar{\pi}(p,\Lambda_{X}(m,\delta))$.
\end{proof}
\bigskip{}

$\textsc{step 3}.$ We conclude the proof by showing that the solution
to (ii'')-(iii'') is unique and minimizes the average weighted cost
function.\bigskip{}
\bigskip{}

\begin{lem}
\label{lemm:minAC}There is a unique $(p^{e},m^{e})$ satisfying conditions
(ii'')-(iii''), and it is characterized by
\[
\{(p^{e},m^{e})\}=\arg\min_{p',m'}\bar{AC}(p',\Lambda(m',\delta)).
\]
\end{lem}
\begin{proof}
See the Appendix.
\end{proof}
\bigskip{}

Figure \ref{fig:E=000026X} describes the intuition behind Lemma \ref{lemm:minAC}.
The pair $(p^{e},m^{e})$ that solves (ii'')-(iii'') is given by the
intersection of the zero entry-profit schedule $\bar{\pi}(p,w,\Lambda(m,\delta))=\kappa$
and the zero exit-profit schedule $\bar{\pi}(p,w,\Lambda_{X}(m,\delta))=0$
in the $(p,m)$ space. By a simple generalization of the textbook
model, the former equation is equivalent to the condition that $p=\bar{AC}(p,w,\Lambda(m,\delta))=\min_{p'}\bar{AC}(p^{'},w,\Lambda(m,\delta))$;
denote the solution to this equation by $\hat{p}(m)$. As illustrated
by the figure, it is also the case that the zero exit-profit schedule
intersects the zero entry-profit schedule at the minimum point of
the latter. Thus, $m^{e}$ minimizes $\bar{AC}(\hat{p}(m),\Lambda(m,\delta))$.
In other words, $(p^{e},m^{e})$ jointly minimize $\bar{AC}$, as
stated in Lemma \ref{lemm:minAC}.

\begin{figure}
\begin{centering}
	\newpsobject{showgrid}{psgrid}{subgriddiv=1,griddots=10,gridlabels=7pt,gridcolor=red} 
	\psset{unit=6mm} 
	\psset{algebraic=true} 
	\begin{pspicture}(0,0)(10,10) 

\psline[linewidth=0.05]{->}(0,0)(10,0) 	
\psline[linewidth=0.05]{->}(0,0)(0,10) 	
\rput(11,-0.5){$m$ (exit threshold)} 	
\rput(-1.3,9){$p$ (price)}	 

\pscurve(0,2)(4,4)(8,10) 

\pscurve(1,7)(3.7,4)(9,7) 

\psline[linestyle=dashed](0,4)(4,4) 
	\psline[linestyle=dashed](4,0)(4,4) 

\psline[linestyle=dotted](0,4.9)(2.4,4.9) 	
\psline[linestyle=dotted](2.4,0)(2.4,4.9)

\rput(4,-0.5){$m^{e}$} 	
\rput(-0.5,4){$p^{e}$} 
	\rput(2.4,-0.5){$\tilde{m}$} 	
\rput(-0.5,5){$\tilde{p}$} 

	\rput(5,10.1){\small{Exit condition:}}	 
	\rput(5,9.5){\small{$\bar{\pi}(p,\Lambda_{X}(m,\delta))=0$}}	 		
	\rput(11,6.1){\small{Entry condition:}}		
	 	\rput(11,5.5){\small{$\bar{\pi}(p,\Lambda(m,\delta))=\kappa$}} 
		\end{pspicture}
\par\end{centering}
\bigskip{}

\bigskip{}

\caption{\label{fig:E=000026X}Characterization of entry and exit conditions.\protect \\
\emph{\footnotesize{}The figure shows the unique price and exit threshold
$(p^{e},m^{e})$ that simultaneously solves the entry and exit conditions.
This solution is given by the intersection of the zero entry-profit
and the zero exit-profit schedules.}}
\end{figure}

The reason why the schedules in Figure \ref{fig:E=000026X} intersect
at the minimum of the zero entry-profit schedule is as follows. Consider
a point $(\tilde{p},\tilde{m})$ on the zero entry-profit schedule
such that $\tilde{m}<m^{e}$. This point lies above the zero exit-profit
schedule; that is, $\bar{\pi}(\tilde{p},\Lambda_{X}(\tilde{m},\delta))>0$,
and so the marginal type $\tilde{m}$ makes a strictly positive profit.
If the marginal type were slightly increased from $\tilde{m}$ to
$\tilde{m}+\varepsilon$, then a potential entrant would stay whenever
drawing a type in $(\tilde{m},\tilde{m}+\varepsilon)$. By continuity,
its profit from having a type in the interval would be positive, and
so the firm's ex-ante profit would increase from zero to a strictly
positive number. The price would then need to fall in order to remain
on the zero entry-profit schedule. Thus, the zero entry-profit schedule
is decreasing whenever it is above the zero exit-profit schedule.
By a similar argument, the zero entry-profit is increasing whenever
it is below the zero exit-profit schedule.

\pagebreak{}

\addcontentsline{toc}{section}{References}

\bibliographystyle{aer}
\bibliography{bibtex}

\appendix

\section{Appendix}

\subsection{Proof of Lemma \ref{lemma_Phi}}

Let $L(\mathcal{M}(\Theta))$ denote the space of linear bounded operators
mapping $\mathcal{M}(\Theta)$ to itself. (i) Since $\varrho||\Phi_{m}||<1$
(here $||.||$ is the operator norm\footnote{The space $\mathcal{M}(\Theta)$ is equipped with the total variation
norm and the operator norm $||\Phi_{m}||\equiv\sup_{\eta\ne0}\frac{||\Phi_{m}[\eta]||_{TV}}{||\eta||_{TV}}\leq1$
where $||\eta||_{TV}\equiv0.5\int_{\Theta}|f_{\eta}(\theta)|d\theta$
where $f_{\eta}$ is the Radon-Nikodym derivative of $\eta$ with
respect to Lebesgue.}), it is easy to see that the sequence $(\sum_{j=0}^{n}\varrho^{j}\Phi_{m}^{j})_{n}$
is Cauchy (under the operator norm). Because $L(\mathcal{M}(\Theta))$
is complete, then $S\equiv\sum_{j=0}^{\infty}\varrho^{j}\Phi_{m}^{j}\in L(\mathcal{M}(\Theta))$.
It is easy to see that $\varrho\Phi_{m}S=S-I$ or, equivalently, $(I-\varrho\Phi_{m})S=I$;
similarly $S(I-\varrho\Phi_{m})=I$. Therefore, $S$ is the inverse
of $(I-\varrho\Phi_{m})$, denoted by $(I-\varrho\Phi_{m})^{-1}$.
(ii) Let $g\in L^{\infty}(\Theta)$ and let $\eta$ be any Borel measure
of $\Theta$. By Fubini's Theorem,
\begin{equation}
\int_{\Theta}T_{m}[g](\theta)\eta(d\theta)=\int_{\Theta}g(\theta')\left\{ \int1\{\theta\leq m\}F(d\theta'\mid\theta)\eta(d\theta)\right\} =\int_{\Theta}g(\theta')\Phi_{m}[\eta](d\theta').\label{eq:Fubini}
\end{equation}
Expression (\ref{eq:Fubini}) can be equivalently be cast as $\langle T_{m}\left[g\right],\eta\rangle=\langle g,\Phi_{m}\left[\eta\right]\rangle$,
where $\langle.,.\rangle$ denotes the integral operation. Using this
notation, it is easy to see that, for any $j$, 
\[
\langle T_{m}^{j}\left[g\right],\eta\rangle=\langle T_{m}\left[T_{m}^{j-1}\left[g\right]\right],\eta\rangle=\langle\left[T_{m}^{j-1}\left[g\right]\right],\Phi_{m}\left[\eta\right]\rangle=...=\langle g,\Phi_{m}^{j}\left[\eta\right]\rangle.\,\,\,\,\square
\]

\subsection{Proof of Lemma \ref{lemm:minAC}}

Throughout this proof, we use the following properties for $V_{m}$.
The proof of these properties follow from standard fixed point arguments
and are thus omitted: (1) For any $m\in\Theta$, $p\mapsto M[V_{m}(p,.)](m)$
is nondecreasing and increasing over $p$ such that $q(p,m)>0$; (2)
For any $m\in\Theta$, $\theta\mapsto M[V_{m}(p,.)](\theta)$ is decreasing;
(3) For any $m\in\Theta$, $p\mapsto M[V_{m}(p,\cdot)](m)$ is continuous.

Before proving Lemma \textbf{\ref{lemm:minAC}}, we state and prove
two preliminary results.

\bigskip{}

\begin{lem}
\label{lem:MVm-shape}For any $p>0$ and any $m\in\Theta$ such that
$\bar{\pi}(p,\Lambda_{X}(m))=0$, \textup{$M[V_{m'}(p,.)-V_{m}(p,.)](\theta)<0$
for all $m'\ne m$ and $\theta\in\Theta$.}
\end{lem}
\begin{proof}
Fix any $\theta\in\Theta.$ We first show the result for $m'<m$.
By definition of $V_{m}(p,.)$,
\begin{align*}
M[V_{m'}(p,.)-V_{m}(p,.)](\theta)= & \varrho\int\left\{ 1\{\theta'\leq m'\}M[V_{m'}(p,.)]-1\{\theta'\leq m\}M[V_{m}(p,.)](\theta')\right\} F(d\theta'\mid\theta)\\
= & \varrho\int1\{m\leq\theta'\leq m'\}M[V_{m}(p,.)](\theta')F(d\theta'\mid\theta)\\
 & +\varrho\int1\{\theta'\leq m'\}M[V_{m'}(p,.)-V_{m}(p,.)](\theta')F(d\theta'\mid\theta)\\
\equiv & A_{m',m}(\theta)+\varrho K_{m'}\left[M[V_{m'}(p,.)-V_{m}(p,.)]\right](\theta),
\end{align*}
where $K_{m'}:L^{\infty}(\Theta)\rightarrow L^{\infty}(\Theta)$ is
given by $K_{m'}[g](\theta)=\int1\{\theta'\leq m'\}g(\theta')F(d\theta'\mid\theta)$.

Observe that $M[V_{m}(p,.)](m)=\bar{\pi}(p,\Lambda_{X}(m))=0$ and
also $\theta\mapsto M[V_{m}(p,.)](\theta)$ is decreasing, hence $1\{m\leq\theta\leq m'\}M[V_{m}(p,.)](\theta)<0$,
which implies $A_{m',m}(.)<0$ (note that $F(\cdot\mid\theta)$ has
full support for all $\theta\in\Theta$ by Assumption \ref{assu:CDFs_cont}(ii)).
Since $\varrho||K_{m'}||=\varrho\sup_{g\in L^{\infty}(\Theta)}\frac{||K_{m'}[g]||_{L^{\infty}(\Theta)}}{||g||_{L^{\infty}(\Theta)}}\leq\varrho<1$,
by the analogous arguments in the proof of Lemma \ref{lemma_Phi},
\[
M[V_{m'}(p,.)-V_{m}(p,.)](\theta)=\left(I-\varrho K_{m'}\right)^{-1}\left[A_{m',m}\right](\theta)=\sum_{j=0}^{\infty}\varrho^{j}K_{m'}^{j}\left[A_{m',m}\right](\theta).
\]

We note that for any $g(.)<0$, $K_{m'}\left[g\right](.)=\int1\{\theta'\leq m'\}g(\theta')F(d\theta'\mid.)<0$.
Hence, from this fact and the fact that $A_{m',m}(.)<0$, we can show
inductively that for each $j$, $\varrho^{j}K_{m'}^{j}\left[A_{m',m}\right](.)$
and thus $M[V_{m'}(p,.)-V_{m}(p,.)](\theta)<0$. 

We now show the case for $m'>m$. Following the same steps as those
above one obtains
\begin{align*}
M[V_{m'}(p,.)-V_{m}(p,.)](\theta)= & -A_{m,m'}(\theta)+\varrho K_{m'}\left[M[V_{m'}(p,.)-V_{m}(p,.)]\right](\theta).
\end{align*}

Since $A_{m,m'}(\theta)=1\{m'\leq\theta\leq m\}M[V_{m}(p,.)](\theta)$,
it follows that $A_{m,m'}(\theta)>0$. This observation and analogous
derivations to the ones for $m'<m$ imply that $M[V_{m'}(p,.)-V_{m}(p,.)](\theta)<0$.
\end{proof}
\begin{lem}
\label{lem:pi-bar-cont}$(p,m)\mapsto\bar{\pi}(p,\Lambda(m))$ and
$(p,m)\mapsto\bar{\pi}(p,\Lambda_{X}(m))$ are continuous.
\end{lem}
\begin{proof}
We only prove continuity of $(p,m)\mapsto\bar{\pi}(p,\Lambda_{X}(m))$
since continuity of $(p,m)\mapsto\bar{\pi}(p,\Lambda(m))$ is obtained
by an analogous argument. By definition of $V_{m}$, we want to show
that $(p,m)\mapsto M[V_{m}(p,\cdot)](m)$ is continuous. Let $(p_{n},m_{n})\rightarrow(p,m)$
and note that, for sufficiently large $n$,
\begin{align*}
|M[V_{m_{n}}(p_{n},\cdot)](m_{n})-M[V_{m}(p,\cdot)](m)|\leq & |M[V_{m_{n}}(p_{n},\cdot)](m_{n})-M[V_{m}(p_{n},\cdot)](m_{n})|\\
 & +|M[V_{m}(p_{n},\cdot)](m_{n})-M[V_{m}(p,\cdot)](m)|\\
\leq & \sup_{p\in C}||M[V_{m_{n}}(p,.)-V_{m}(p,\cdot)]||_{L^{\infty}}\\
 & +|M[V_{m}(p_{n},\cdot)](m_{n})-M[V_{m}(p,\cdot)](m)|
\end{align*}
where $C$ is some compact neighborhood of $p$. The second term in
the RHS vanishes because $(p,t)\mapsto M[V_{m}(p,.)](t)$ is continuous
(the proof follows from standard contraction mapping arguments and
is omitted). Thus, the desired result follows by showing that the
first term in the RHS vanish. To do this, note that for any $\theta\in\Theta$
and any $p\in C$,
\begin{align*}
|M[V_{m_{n}}(p,\cdot)-V_{m}(p,\cdot)](\theta)|\leq & \varrho|\int\left(1\{\theta\leq m_{n}\}-1\{\theta\leq m\}\right)M[V_{m_{n}}(p,.)](\theta')f(\theta'\mid\theta)d\theta'|\\
 & +\varrho|\int\left(1\{\theta\leq m\}\right)M[V_{m_{n}}(p,.)-V_{m}(p,\cdot)](\theta')f(\theta'\mid\theta)d\theta'|\\
\leq & \varrho|B_{m_{n},m,p}(\theta)|+\varrho||M[V_{m_{n}}(p,.)-V_{m}(p,\cdot)]||_{L^{\infty}}.
\end{align*}
where $B_{m_{n},m}(\theta)\equiv\int\left(1\{\theta\leq m_{n}\}-1\{\theta\leq m\}\right)M[V_{m'}(p,.)](\theta')f(\theta'\mid\theta)d\theta'$.
Therefore, since $\varrho<1,$ it suffices to show that there exists
a $\delta>0$ such that $\limsup_{n\rightarrow\infty}\sup_{p\in C}||B_{m_{n},m,p}||_{L^{\infty}}=0$.
To do this, we first show that for each $\theta$, $\limsup_{n\rightarrow\infty}\sup_{p\in C}|B_{m_{n},m,p}(\theta)|=0$. 

It is easy to show that there exists a $K<\infty$ such that $\sup_{p\in C}\sup_{m\in\Theta}||V_{m}(p,\cdot)||_{L^{\infty}}\leq K$.
So, for any $\theta'\in\Theta$, 
\[
\sup_{p\in C}|\left(1\{\theta\leq m_{n}\}-1\{\theta\leq m\}\right)M[V_{m_{n}}(p,.)](\theta')f(\theta'\mid\theta)|\leq K|\left(1\{\theta\leq m_{n}\}-1\{\theta\leq m\}\right)f(\theta'\mid\theta)|.
\]
Thus, for any $\theta'\ne m$, $\limsup_{n\rightarrow\infty}\sup_{p\in C}|\left(1\{\theta\leq m_{n}\}-1\{\theta\leq m\}\right)M[V_{m_{n}}(p,.)](\theta')f(\theta'\mid\theta)|=0$.
By the DCT, this readily implies that for any $\theta\in\Theta$,
$\limsup_{n\rightarrow\infty}\sup_{p\in C}|B_{m_{n},m,p}(\theta)|=0$. 

We now show that $\limsup_{n\rightarrow\infty}\sup_{\theta\in\Theta}\sup_{p\in C}|B_{m_{n},m,p}(\theta)|=0$.
Since $\Theta$ is compact and we already established pointwise convergence,
by the Arzela-Ascoli theorem it suffices to show that the family $\{\sup_{p\in C}|B_{m_{n},m,p}(\cdot)|\}_{n\in\mathbb{N}}$
is equi-continuous. To do this, note that for any $\theta$ and $\theta'$,{\small{}\noindent
\begin{align*}
\sup_{p\in C}|B_{m_{n},m,p}(\theta')|-\sup_{p\in C}|B_{m_{n},m,p}(\theta)|\leq & \sup_{p\in C}\left\{ |B_{m_{n},m,p}(\theta')|-|B_{m_{n},m,p}(\theta)|\right\} \\
\leq & \sup_{p\in C}|\int(1\{\theta\leq m_{n}\}-1\{\theta\leq m\})M[V_{m_{n}}(p,.)](t)\left(f(t\mid\theta)-f(t\mid\theta')\right)dt|\\
\leq & K\times|\left(f(t\mid\theta)-f(t\mid\theta')\right)dt|.
\end{align*}
}The RHS is continuous by Assumption \ref{assu:CDFs_cont}(ii), and
its ``modulus of continuity'' does not depend on $m'$. Hence, $\{\sup_{p\in C}|B_{m',m,p}(\cdot)|\}_{n\in\mathbb{N}}$
is equi-continuous.
\end{proof}
\bigskip{}

\textbf{Proof of Lemma \ref{lemm:minAC}. }Throughout the proof, we
fix $\delta$ and omit it from the notation. We now define certain
mappings that will be used throughout the proof. Let $m\mapsto p_{E}(m)\equiv\{p\colon\bar{\pi}(p,\Lambda(m))=\kappa\}$,
and $p\mapsto m_{X}(p)\equiv\{m\colon\bar{\pi}(p,\Lambda_{X}(m))=0\}$
and $m\mapsto p_{X}(m)=\{p\colon m_{X}(p)=m\}$. For the mapping $m_{X}$,
it is implicit that if $\bar{\pi}(p,\Lambda_{X}(m))<0$ then $m_{X}(p)=\theta_{L}$
and if $\bar{\pi}(p,\Lambda_{X}(m))>0$ then $m_{X}(p)=\theta_{H}$.
\bigskip{}

STEP 1. We now show that a solution to the system (ii')-(iii') exists
and is unique and, moreover, we show that for any $(m,p)$ such that
$\bar{\pi}(p,\Lambda_{X}(m))=0$ and $\bar{\pi}(p,\Lambda(m))=\kappa$,
then $p<p_{E}(m')$ for all $m'\ne m$, i.e., $m$ is a global minimizer
of the function $p_{E}$. 

Observe that by Assumption \ref{assu:CDFs_cont}(i), $\nu\left(\left\{ C(0,\theta)>0\right\} \right)>0$.
Also, $supp(\Lambda(m))\supseteq supp(\nu)$ for all $m$, so $\int C(0,\theta)\Lambda(m)(d\theta)>0$.
This implies that if $\bar{q}(p,\Lambda(m))=0$, then $\bar{\pi}(p,\Lambda(m))<0\leq\kappa$,
so a $(p,m)$ such that $\bar{q}(p,\Lambda(m))=0$ can never be a
solution to $\bar{\pi}(p,\Lambda(m))=\kappa$ (if it exists). Therefore,
if the solution exists it would be such that $\bar{q}(p,\Lambda(m))>0$,
in particular, this implies that $p=0$ cannot be part of a solution.
Therefore, henceforth we focus on $(p,m)$ such that $\Lambda(m)(\{\theta\colon q(p,w,\theta)>0\})>0$,
in particular, we only consider $m\in M\equiv\{m\in\Theta\colon\exists p\colon\Lambda(m)(\{\theta\colon q(p,\theta)>0\})>0\}$.

One of the following cases occurs: (a) $p_{E}-p_{X}<0$; (b) $p_{E}-p_{X}>0$
or (c) neither (a) nor (b) occurs (i.e., $p_{E}-p_{X}$ changes signs
at least once in $\Theta$). If (a) occurs, then the solution to (ii')-(iii')
exists and is given by $m=\theta_{L}$ and $p$ such that $\bar{\pi}(p,\Lambda(m))=\kappa$
and $\bar{\pi}(m,\Lambda_{X}(m))<0$. Similarly, if (b) occurs, then
the solution to (ii')-(iii') exists and is given by $m=\theta_{H}$
and $p$ such that $\bar{\pi}(p,\Lambda(m))=\kappa$ and $\bar{\pi}(m,\Lambda_{X}(m))>0$.
Therefore, if either (a) or (b) occurs a solution exists and is unique.

We now show that the same holds if (c) occurs. Clearly, for existence
of a solution in this case it suffices that $m\mapsto p_{X}(m)$ is
continuous (i.e., for any $(m_{n})_{n}$ and $(p_{n})_{n}$ such that
$m_{n}\rightarrow m$ and $p_{n}\in p_{X}(m_{n})$ with $p_{n}\rightarrow p$
then $p\in p_{X}(m)$) and closed- and convex-valued; and that $m\mapsto p_{E}(m)$
is single-valued and continuous. Continuity of $m\mapsto p_{X}(m)$
follows from Lemma \ref{lem:pi-bar-cont}; and by continuity and monotonicity
of $p\mapsto M[V_{m}(p,\cdot)](m)$, it follows that, for each $m\in\Theta$,
$p_{X}(m)$ is a closed interval. Since $p\mapsto\pi(p,\theta)$ is
nondecreasing and increasing over $p$ such that $q(p,\theta)>0$
and $supp(\Lambda(m))\supseteq supp(\nu)$ for all $m$, it follows
that for any $m\in M$, $p\mapsto\bar{\pi}(p,\Lambda(m))$ is increasing.
Hence, $p_{E}(m)$ has at most one element. Moreover, since $\bar{\pi}(0,\Lambda(m))\leq0$
and $\lim\inf_{p\rightarrow\infty}\bar{\pi}(p,\Lambda(m))=\infty$,
continuity of $p\mapsto\bar{\pi}(p,\Lambda(m))$ ensures that $p_{E}(m)$
is non-empty. Finally, continuity of $m\mapsto p_{E}(m)$ follows
from Lemma \ref{lem:pi-bar-cont}.

It thus remains to show that the solution in case (c) is unique. To
do this, it suffices to show that for any $(m,p)$ such that $\bar{\pi}(p,\Lambda_{X}(m))=0$
and $\bar{\pi}(p,\Lambda(m))=\kappa$, then $p<p_{E}(m')$ for all
$m'\ne m$, i.e., $m$ is a global minimizer of the function $p_{E}$.
Since $p\mapsto\bar{\pi}(p,\Lambda(m))$ is increasing, it suffices
to show that for any $m'\ne m$, $\bar{\pi}(p,\Lambda(m'))<\bar{\pi}(p,\Lambda(m))=\kappa$. 

For any $m_{1}\leq m_{2}$, let $\theta\mapsto A_{m_{1},m_{2}}(\theta)\equiv1\{m_{1}\leq\theta\leq m_{2}\}M[V_{m}(p,.)](\theta)$.
Note that $M[V_{m}(p,.)](m)=\bar{\pi}(p,\Lambda_{X}(m))=0$ and also
$\theta\mapsto M[V_{m}(p,.)](\theta)$ is decreasing, so $M[V_{m}(p,.)](.)<(>)0$
for all $\theta>(<)m$. This, in turn, implies that $A_{m_{1},m}(.)>0$
and $A_{m,m_{2}}(.)<0$. 

By definition of $V_{m}$, it follows that: If $m'>m$, 
\[
\bar{\pi}(p,\Lambda(m'))-\bar{\pi}(p,\Lambda(m))=\int A_{m,m'}(\theta)\nu(d\theta)+\varrho\int1\{\theta\leq m'\}M[V_{m'}(p,.)-V_{m}(p,.)](\theta)\nu(d\theta)
\]
and if $m'<m$, 
\[
\bar{\pi}(p,\Lambda(m'))-\bar{\pi}(p,\Lambda(m))=-\int A_{m',m}(\theta)\nu(d\theta)+\varrho\int1\{\theta\leq m'\}M[V_{m'}(p,.)-V_{m}(p,.)](\theta)\nu(d\theta).
\]

By our previous observations, $\int A_{m,m'}(\theta)\nu(d\theta)<0$
and $-\int A_{m',m}(\theta)\nu(d\theta)<0$. By Lemma \ref{lem:MVm-shape}
$M[V_{m'}(p,.)-V_{m}(p,.)](\theta)<0$ for any $m'\ne m$ and any
$\theta\in\Theta$. So, $\bar{\pi}(p,\Lambda(m'))-\bar{\pi}(p,\Lambda(m))=\bar{\pi}(p,\Lambda(m'))-\kappa<0$
as desired. 

\bigskip{}

STEP 2. We now show that the solution $(p^{e},m^{e})$ for (ii')-(iii'),
which is unique (see Step 1), satisfies 
\[
(p^{e},m^{e})=\arg\min_{p',m'}\bar{AC}(p',\Lambda(m',\delta)).
\]
 To show this, we first show that for any $m$, $\bar{\pi}(p,\Lambda(m))=\kappa$
iff $p=\bar{AC}(p,\Lambda(m))$ iff $p=\min_{p'\geq0}AC(p',\Lambda(m))$.
The first `iff' follows from simple algebra. To show the second `iff',
let $p_{m}\equiv\inf_{p:\bar{q}(p,\Lambda(m))=0}p<\infty$, and that
implies that $\bar{AC}(p_{m},\Lambda(m))=\infty$ for all $p\leq p_{m}$.
Suppose for now (we show it below) that the following holds: (I) If
$p<\bar{AC}(p,\Lambda(m))$, then $\bar{AC}(p',\Lambda(m))<\bar{AC}(p,\Lambda(m))$
for all $p'$ such that $p<p'<\bar{AC}(p,\Lambda(m))$; (II) If $p>\bar{AC}(p,\Lambda(m))$,
then $\bar{AC}(p',\Lambda(m))>\bar{AC}(p,\Lambda(m))$ for all $p'>p$;
and (III) There is at most one solution $p$ to $p=\bar{AC}(p,\Lambda(m))$. 

We claim that by (I) and the facts that $\bar{AC}(p_{m},\Lambda(m))=\infty$
for all $p\leq p_{m}$ and continuity of $\bar{AC}(\cdot,\Lambda(m))$
over $p>p_{m}$, there exists a solution $p$ to $p=\bar{AC}(p,\Lambda(m))$
and $p>p_{m}$. To show this, suppose not, i.e., $p<\bar{AC}(p,\Lambda(m))$
for all $p$. This implies that there exists a $p'$ such that $\bar{AC}(p',\Lambda(m))<\bar{AC}(p'',\Lambda(m))$
for all $p''\ne p'$, in particular for any $p'<p''<\bar{AC}(p',\Lambda(m))$.
But this contradicts (I). By (I) and (II), this solution minimizes
$\bar{AC}(\cdot,\Lambda(m))$, and, by (III), this is the unique solution. 

We now prove (I)-(III). Let $p,p'>p_{m}$ and $p'>p$. By definition
of optimality, $pq(p,\theta)-C(q(p,\theta),\theta)\geq pq(p',\theta)-C(q(p',\theta),\theta)$
and $p'q(p',\theta)-C(q(p',\theta),\theta)\geq p'q(p,\theta)-C(q(p,\theta),\theta)$.
By simple algebra, integrating over $\Theta$ using $\Lambda(m)(.)$,
and the fact that $\bar{q}(p',\Lambda(m))-\bar{q}(p,\Lambda(m))>0$
(by the assumption that $p'>p>p_{m}$ and the fact that $p\mapsto q(p,\theta)$
is increasing for over $p$ such that $q(p,\theta)>0$),
\begin{equation}
p\leq\frac{\bar{C}(p',\Lambda(m))-\bar{C}(p,\Lambda(m))}{\bar{q}(p',\Lambda(m))-\bar{q}(p,\Lambda(m))}\leq p'.\label{eq:MC_discrete}
\end{equation}
First, suppose that $p<\bar{AC}(p,\Lambda(m))$. Then (\ref{eq:MC_discrete})
implies that, for all $p'$ such that $p<p'<\bar{AC}(p,\Lambda(m))$,
\[
\bar{AC}(p',\Lambda(m))\equiv\frac{\bar{C}(p',\Lambda(m))}{\bar{q}(p',\Lambda(m))}<\frac{\bar{C}(p,\Lambda(m))}{\bar{q}(p,\Lambda(m))}\equiv\bar{AC}(p,\Lambda(m)).
\]
Thus, (I) is proven. Next, let $p>\bar{AC}(p,\Lambda(m))$. Then (\ref{eq:MC_discrete})
implies that $\bar{AC}(p',\Lambda(m))>\bar{AC}(p,\Lambda(m))$ for
all $p'>p$; thus, (II) is proven. Finally, suppose $p=\bar{AC}(p,\Lambda(m))$
and $p'=\bar{AC}(p',\Lambda(m))$ with $p'>p$. Putting together the
two inequalities in (\ref{eq:MC_discrete}), $p=\bar{AC}(p,\Lambda(m))=\bar{AC}(p',\Lambda(m))=p'$,
which contradicts $p'>p$. A similar contradiction obtains if we assume
$p'<p$. Therefore, $p'=p$, and so (III) is proven. 

Note that $p_{E}(m)=\arg\min_{p'\geq0}\bar{AC}(p',\Lambda(m))$. Moreover,
if $(m^{e},p^{e})$ solves (ii')-(iii'), $p^{e}=p_{E}(m^{e})$. So
in order to show the desired result it suffices to show that $m^{e}=\arg\min_{m\in\Theta}\bar{AC}(p_{E}(m),\Lambda(m))$,
or equivalently, $\bar{AC}(p_{E}(m^{e}),\Lambda(m^{e}))<\bar{AC}(p_{E}(m),\Lambda(m))$
for all $m\ne m^{e}$. By step 1,
\[
p^{e}=p_{E}(m^{e})<p_{E}(m')
\]
for all $m'\ne m$. Since, by our previous calculations in this step,
$\bar{AC}(p_{E}(m),\Lambda(m))=p_{E}(m)$ for all $m$, the desired
result follows. $\square$

\bigskip{}

\end{document}